\documentclass[sigconf]{acmart}

\usepackage{graphicx}
\usepackage{import}
\usepackage{xcolor}
\usepackage{mathtools}
\usepackage{bbm}
\usepackage{amsmath}
\usepackage{float}

\newtheorem{Lemma}{Lemma}

\usepackage[linesnumbered,algoruled,boxed,lined, noend, noline]{algorithm2e}

\newcommand{\trival}[1]{\Delta_{#1}}
\newcommand{\btri}[1]{B\Delta^{#1}}
\newcommand{\utri}[1]{U\Delta^{#1}}
\newcommand{\btrival}[2]{B\Delta^{#1}_{#2}}
\newcommand{\utrival}[2]{U\Delta^{#1}_{#2}}
\newcommand{\ctrival}[2]{?\Delta^{#1}_{#2}}
\AtBeginDocument{%
  }

\begin{document}

\title{Counting Balanced Triangles on Social Networks With Uncertain Edge Signs}

\author{Alexander Zhou}
\affiliation{%
  \institution{Hong Kong Polytechnic University}
}
\email{alexander.zhou@polyu.edu.hk}

\author{Haoyang Li}
\affiliation{%
  \institution{Hong Kong Polytechnic University}
}
\email{haoyang-comp.li@polyu.edu.hk}

\author{Anxin Tian}
\affiliation{%
  \institution{Hong Kong University of Science and Technology}
}
\email{atian@connect.ust.hk}

\author{Zhiyuan Li}
\affiliation{%
  \institution{Hong Kong University of Science and Technology}
}
\email{zlicw@cse.ust.hk}

\author{Yue Wang}
\affiliation{%
  \institution{Shenzhen Institute of Computing Sciences}
}
\email{yuewang@sics.ac.cn}

\begin{abstract}
On signed social networks, balanced and unbalanced triangles are a critical motif due to their role as the foundations of Structural Balance Theory. The uses for these motifs have been extensively explored in networks with known edge signs, however in the real-world graphs with ground-truth signs are near non-existent, particularly on a large-scale. In reality, edge signs are inferred via various techniques with differing levels of confidence, meaning the edge signs on these graphs should be modelled with a probability value. In this work, we adapt balanced and unbalanced triangles to a setting with uncertain edge signs and explore the problems of triangle counting and enumeration. We provide a baseline and improved method (leveraging the inherent information provided by the edge probabilities in order to reduce the search space) for fast exact counting and enumeration. We also explore approximate solutions for counting via different sampling approaches, including leveraging insights from our improved exact solution to significantly reduce the runtime of each sample resulting in upwards of two magnitudes more queries executed per second. We evaluate the efficiency of all our solutions as well as examine the effectiveness of our sampling approaches on real-world topological networks with a variety of probability distributions.
\end{abstract}

\maketitle

\section{Introduction}
The triangle, a 3-edge and 3-node cycle, has been extensively utilised as a core motif for analysis and subgraph discovery in signed social networks, where edges between users exist with a positive (+) or negative (-) sign. This model is commonly utilised when examining online social networks as a positive edge can indicate friendship, trust or similarity whilst a negative edge is more indicative of antagonism, distrust or simply just difference \cite{signed/leskovec, signed/tang}. The signed network is also utilised in many other fields such as political science \cite{signed/usecase/fontan}, neuroscience \cite{signed/usecase/saberi} and ecology \cite{signed/usecase/saiz}.

\begin{figure}
\centering
\includegraphics[height = 2.3cm]{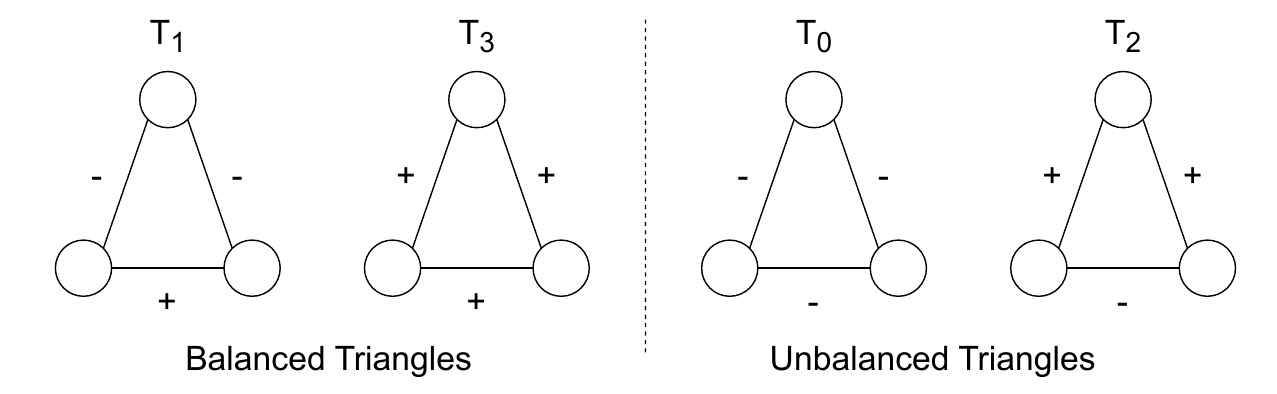}
\caption{An example of balanced and unbalanced triangles.}
\label{fig/baltri}
\end{figure}

One reason that triangles are a popular motif to research on these signed graphs is due to their role in `Structural Balance Theory', a concept with roots in social science, which is based on the idea that ``The enemy of my enemy is my friend'' \cite{signed/balanced/davis, signed/balanced/harary}. A balanced triangle is one which consists of an odd number of `+' edges, whilst an unbalanced triangle consists of an even number of '+' signs. Figure \ref{fig/baltri} illustrates the four possible types of triangles in a signed network and separates them into balanced and unbalanced classes.

Established research problems that consider Structural Balance Theory on signed social networks includes triangle counting \cite{triangle/signed/arya}, spam classification \cite{signed/usecase/jeong}, balance maximization \cite{signed/balance/sharma}, balanced subgraph \cite{signed/balanced/bhore, signed/balanced/chen2, signed/balanced/figueiredo, signed/balanced/ordozgoiti} and community search \cite{signed/balance/anchuri} including classic cohesive subgraph structures such as cliques \cite{signed/balanced/chen, signed/balanced/yao}, cores \cite{signed/balanced/sun} and trusses \cite{signed/balanced/wu, signed/balanced/zhao}. Communities based on structural balance theory are particularly of interest as they inherently produce Polarized Communities which consist of individual Echo Chambers which only connect to other communities via negative connections \cite{signed/balanced/chen, signed/balanced/yao, me/epc}.

However one crucial detail which is almost entirely ignored when discussing signed networks is the acknowledgement that in real-world networks each sign that has been assigned to an edge is near-impossible to know with certainty. Aside from smaller available datasets such as Epinions and Slashdot \footnote{Available on Stanford SNAP: http://snap.stanford.edu/data/index.html} in which the signs are ground truth labels given by users (even then there still exist potential questions regarding truthfulness or noise), there are no publicly available networks with ground truth signs. Indeed there is a subsection of research aiming to predict signs via some inference technique \cite{prediction/javari, prediction/khodadadi, prediction/leskovec, prediction/tang}. This inherently limits the real world usability of all signed graph algorithms as they cannot handle edges in which the signs are not certain.

\begin{figure}
\centering
\includegraphics[height = 2.5cm]{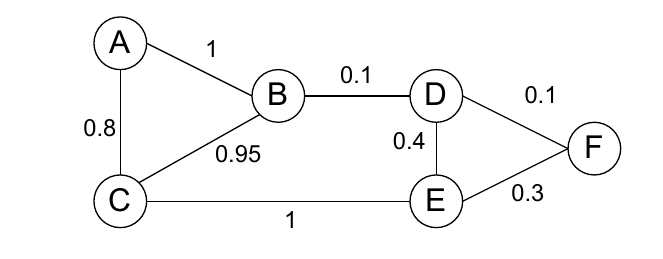}
\caption{An example of a signed graph with uncertain edges. The number associated with the edge is the probability the edge is positive.}
\label{fig/unsigngraphex}
\end{figure}

Given that signs are so inherently uncertain for the purposes of real-world networks, it should naturally follow that research on signed networks should account for the probability associated with the signs. In this work we explore this problem from the direction of balanced and unbalanced triangles. 

Take the example graph in Figure \ref{fig/unsigngraphex}, where the number on each edge refers to the probability that edge is positive. The triangles formed by the nodes $A, B, C$ has a 0.77 probability of being balanced whilst the triangle formed by nodes $D, E, F$ has only a 0.468 probability. These probabilities indicate different semantic information, with the $A, B, C$ triangle being much more likely to be balanced whilst it is essentially inconclusive whether the $D, E, F$ is balanced or unbalanced. 

Our work proposes an efficient algorithm to count and/or enumerate all triangles in the network which are likely to be balanced or unbalanced (with probability exceeding some threshold value $t$). Given that triangle counting for extremely large graphs is computationally expensive, we also propose sampling approaches which allows for the network operator to estimate the count of the uncertain balanced and unbalanced triangles in the network quickly.

The reason that being able to the find triangles that are likely to be balanced (or unbalanced) on a network with uncertain edge signs is interesting is because it allows us to reap all the described benefits of balanced triangles on a more realistically obtainable network (given ground truth signs are so elusive). By definition, all these problems require the discovery of balanced (or unbalanced) triangles and our method would allow the further exploration of these problems on networks with inferred signs (where certainty is not guaranteed) \cite{signed/balanced/chen, signed/balanced/chen2, signed/balanced/partition/doreian, signed/balanced/figueiredo, signed/balanced/partition/krauthgamer, signed/balanced/link/liu, signed/balanced/ordozgoiti, signed/balanced/sun, signed/balanced/yao}. We propose the first work to consider triangle counting and Structural Balance Theory in a signed network in which the signs are uncertain. 

To outline the direction of this work, we first review the existing literature on triangle counting as well as research dedicated to signed networks with a particular focus on those which are built around Balance Theory (Section 2). We then formally define our problem (Section 3) before introducing a baseline exact counting method (Section 4). We then introduce an improved counting algorithm which leverages the inherent properties of probability theory (Section 5) before introducing our approximate sampling methods (Section 6). We verify the performance and accuracy of our methods against a 12 networks (Section 7) before discussing future directions and concluding this work (Section 8).

\section{Related Works}
Exact and approximate triangle counting on a normal, undirected graph has been extensively studied \cite{triangle/batagelj, triangle/chu, para/green, triangle/hasan, para/hoang, para/hu, para/huSig, triangle/latapy, triangle/luce, triangle/schank}. These established techniques do not take into account signs or probabilities, making them difficult to trivially adapt to our setting.

Triangles are of significant interest in web research as they are used as building blocks of community structures for online and offline social networks, in particular the truss model in which each edge exists only if it is a part of a predefined number of triangles \cite{triangle/truss/akbas, triangle/uncertain/huang, triangle/truss/liao, me/dtruss, signed/balanced/wu, triangle/truss/xu, signed/balanced/zhao, triangle/uncertain/zou}.

\textbf{Signed Graph Problems: }
Signed networks have also recently seen significant attention as a research direction due to how it models rich semantic information on social networks. Non-structural balance theory research on signed graph includes notable problems such as influence maximization \cite{signed/infmax/ju, signed/infmax/liang, signed/infmax/yin}, recommendation \cite{signed/rec/symeonidis, signed/rec/tang} and community detection \cite{signed/nonbalcom/chen, signed/nonbalcom/li, signed/nonbalcom/yang, me/epc}.

Structural balance theory as a foundation for problems on signed graphs is an area of significant interest, in particular as a formulation tool for community structures \cite{signed/balanced/chen, signed/balanced/chen2, signed/balanced/figueiredo, signed/balanced/ordozgoiti, signed/balanced/sun, signed/balanced/yao, signed/balanced/zhao}. Other notable uses for balance theory on signed networks include balance maximization \cite{signed/balance/sharma}, partitioning \cite{signed/balanced/partition/doreian, signed/balanced/partition/krauthgamer}, sign inference \cite{signed/balanced/inference/patidar} and link prediction \cite{signed/balanced/link/chiang, signed/balanced/link/liu}. On the bipartite setting, in which the triangle motif is replaced by the butterfly, structural balance theory has also been explored \cite{signed/bipartite/derr, me/balance}. Similar work has also been explored in finding butterflies in the uncertain graph setting where the edges themselves exist with certain probability \cite{me/ubfc, me/ubitd}.

On signed graphs in particular, the problem of counting triangles is often viewed from the lens of being balanced or unbalanced with one notable solution providing an incremental method which is operational on dynamic systems \cite{triangle/signed/arya}. Due to the nature of exact balanced triangle counting and listing largely requiring an exhaustive examination of all triangles, many existing works just use standard triangle enumeration techniques \cite{signed/balanced/chen, signed/balanced/sun, signed/balanced/yao, signed/balanced/zhao}. There has also been some notable work conducted on motif counting (including triangles) in the uncertain space \cite{triangle/uncertain/ma} as well as a truss-based communities \cite{triangle/uncertain/huang, triangle/uncertain/zou} however the techniques discussed in these works are unsuitable to direct adaptation to our problem. Our work is the first to propose a uniform problem setting for considering triangles under the conditions of probabilistic edge signs, opening up avenues for extending the literature in this new direction. Recently a work was proposed that examines signed uncertain graphs for antagonistic groups, though they additionally require weighted semantic information \cite{sug/zhang}.

\section{Preliminaries}
\begin{table}[ht]
\begin{center}
\caption{Key Notations and Definitions}
\begin{tabular}{ |c|c| } 
\hline
\textbf{Notation} & \textbf{Definition} \\ 
\hline
$G = (V, E, P_+)$ & Uncertain Signed Graph\\
\hline
$deg(u)$ & Degree of Node $u$\\
\hline
$N(u)$ & Neighbour Set of $u$\\
\hline
$P_+(e)$ & Probability the Sign of $e$ is Positive \\
\hline
$P_-(e) = 1 - P_+(e)$ & Probability the Sign of $e$ is Negative\\ \hline
$P_{bal}(.)$ & Probability a Subgraph is Balanced \\
\hline
$P_{unbal}(.) = 1 - P_{bal}(.)$ & Probability a Subgraph is Unbalanced \\
\hline
$t$ & Probability Threshold Value\\
\hline
$\trival{u, v, w}$ & Triangle with Nodes $u$, $v$ and $w$\\
\hline
$\btrival{t}{u, v, w}$ & Uncertain Balanced Triangle \\
\hline
$\utrival{t}{u, v, w}$ & Uncertain Unbalanced Triangle \\
\hline
$\ctrival{t}{u, v, w}$ & Uncertain Unclassified Triangle\\
\hline
$T_i$ & Set of Triangles with $i$ Positive Edges\\
\hline
$<_{no}$ & Node Order\\
\hline
$a(e)$ & Absolute Edge Value\\
\hline
$<_{aeo}$ & Absolute Edge Order\\
\hline
\end{tabular}
\label{deftable}
\end{center}
\end{table}

In this section we formally define the problem of Uncertain Balanced and Unbalanced Triangle Counting and Enumeration on Uncertain Signed Graphs. Table \ref{deftable} provides a summary of the key notation.

\begin{definition}
\textbf{Uncertain Signed Graph: } An uncertain signed graph $G = (V, E, P_+)$ consists of a set of vertices $V$ as well as a set of edges $E$. We use $e(u, v)$ to denote an edge connecting nodes the nodes $u, v \in V$. $P_+$ is a set of mappings such that $P_+(e) \in [0, 1]$ is the probability that an edge $e \in E$ has a positive sign.
\end{definition}

The (undirected) Uncertain Signed Graph setting assigns each edge with the probability that the sign associated with the edge is positive (and thus we can calculate the probability the edge is negative via ($P_-(e) = 1 - P_+(e)$)). This network setting assumes that each edge itself exists with certainty and that only the signs themselves can be uncertain but each edge must have a sign. Edges can also be assigned $P_+$ values of 1 or 0 meaning they are known to be positive or negative with certainty respectively.

\begin{definition}
\textbf{Triangle ($\trival{u,v,w}$): } A triangle $\trival{u, v, w}$ on $G$ consisting of nodes $u$, $v$ and $w$ in $V$ is a fully connected subgraph on these three nodes (i.e. edges $e(u, v)$, $e(u, w)$ and $e(v, w)$ all exist in $E$). 
\end{definition}

\begin{definition}
\textbf{Balanced and Unbalanced Triangles: } A balanced triangle (on signed graphs with certain edge signs) is a triangle which contains an odd number of positive edges. By extension, an unbalanced triangle is a triangle which contains an even number of positive edges. 
\end{definition}

The basic topological structure of the triangle motif on Uncertain Signed Graphs remains unchanged from the standard graph setting. We use $T_i$ to refer to the class of triangles with $i$ positive edges. Figure \ref{fig/baltri} illustrates balanced and unbalanced triangles on the Signed Graph setting (with known edge signs). In essence, triangles with 1 (class $T_1$) or 3 (class $T_3$) positive edges are balanced and triangles with 0 (class $T_0$) or 2 (class $T_2$) positive edges are unbalanced. We now extend the concept of balanced and unbalanced triangles to the Uncertain Signed Graph setting.

\begin{definition}
\textbf{Uncertain Balanced Triangle $\btrival{t}{u, v, w}$:} An uncertain balanced triangle $\btrival{t}{u, v, w} \in G$ is a triangle which is balanced with probability greater or equal to a threshold value $t \in [0, 1]$.
\end{definition}

\begin{definition}
\textbf{Uncertain Unbalanced Triangle $\utrival{t}{u, v, w}$:} An uncertain balanced triangle $\utrival{t}{u, v, w} \in G$ is a triangle which is unbalanced with probability greater than a threshold value $t \in [0, 1]$.
\end{definition}

We introduce a user-defined threshold parameter $t$ which is the probability that must be exceeded for a triangle $\trival{u, v, w}$ to be considered either a balanced or unbalanced triangle on the uncertain signed graph. For practical purposes we assume $t \geq 0.5$. The probability a triangle is balanced can be calculated via: 

\begin{equation*}
\begin{split}
    P_{bal}(\trival{u, v, w}) 
    & = P(\trival{u, v, w} \in T_3) + P(\trival{u, v, w} \in T_1) \\
    & = (P_+(u, v) * P_+(u, w) * P_+(v, w)) \\ 
    & + (P_+(u, v) * P_-(u, w) * P_-(v, w)) \\ 
    & + (P_-(u, v) * P_+(u, w) * P_-(v, w)) \\ 
    & + (P_-(u, v) * P_-(u, w) * P_+(v, w))
\end{split}
\end{equation*}

In essence, the probability a triangle is balanced ($P_{bal}(\trival{u, v, w})$) is equal to the probability that all three signs of the triangle are positive added to the probability that only one of the three edges of the triangle is positive (i.e. the triangle is in class $T_1$ or class $T_3$). Additionally, the probability that a triangle is unbalanced ($P_{unbal}(\trival{u, v, w})$) is the probability that either two or zero signs of the triangle are positive (i.e. the triangle is in class $T_0$ or $T_2$) or $1 - P_{bal}(\trival{u, v, w})$.

In our setting $t$ is the only user parameter, which from a practical perspective is quite reasonable for a network operator to decide upon as it is inherently understandable as a concept (a threshold that a probability has to surpass).

\begin{figure}
\centering
\includegraphics[height = 5cm]{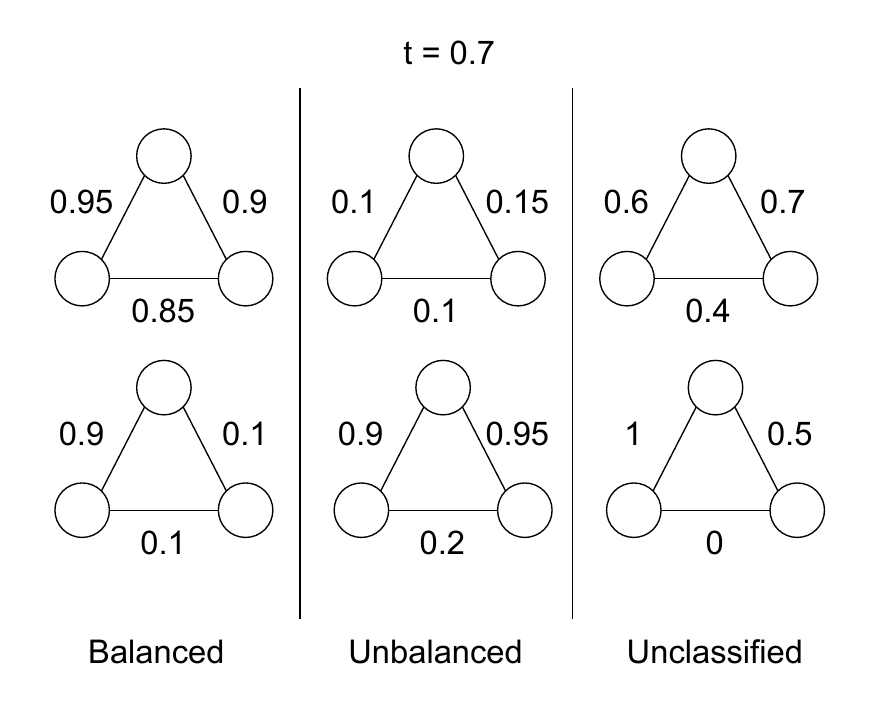}
\caption{An example of uncertain balanced, unbalanced and unclassified triangles for $t = 0.7$. The numerical value of each edge corresponds to $P_+(e)$}
\label{fig/untriex}
\end{figure}

We illustrate some example uncertain balanced and unbalanced triangles in Figure \ref{fig/untriex} for a threshold value of $t = 0.7$. In the left column we have two triangles which are likely to be balanced with probability greater than $0.7$. The top triangle is most likely to belong to the class $T_3$ whilst the bottom triangle is most likely to belong to the class $T_1$. For the middle column we have two triangles which are likely to be unbalanced, with the top being in class $T_0$ and the bottom being in class $T_2$.

This method of classifying triangles as either balanced or unbalanced if they reach a certain threshold means for any value of $t > 0.5$ there can exist triangles which do not satisfy the probability threshold requirement for being balanced or unbalanced. We refer to this set of triangles as Unclassified Triangles.

\begin{definition}
\textbf{Uncertain Unclassified Triangle $\ctrival{t}{u, v, w}$:} An uncertain unclassified triangle $\ctrival{t}{u, v, w} \in G$ is a triangle which is balanced with a probability in the range $[\min\{1 - t, t\}, \max\{1 - t, t\})$ for a threshold value $t \in [0, 1]$.
\end{definition}

Referring back to Figure \ref{fig/untriex}, the right column demonstrates two triangles which are unclassified. The top triangle is one where all three edges are too difficult to distinguish as likely to be positive or negative and the bottom triangle is one in which one edge is a coin-flip and the other two are fixed thus making it equally likely for the triangle to be balanced or unbalanced and thus useless for balanced or unbalance triangle analytics.

\begin{definition}
\textbf{The Uncertain Balanced/Unbalanced Triangle Enumeration Problem: } For an uncertain signed graph $G$ and a probability threshold value $t$, enumerate all Uncertain Balanced and/or Unbalanced triangles in $G$.
\end{definition}

\begin{definition}
\textbf{The Uncertain Balanced/Unbalanced Triangle Counting Problem: } For an uncertain signed graph $G$ and a probability threshold value $t$, determine the number of Uncertain Balanced and/or Unbalanced triangles in $G$.
\end{definition}

These problems are the two most widely explored research problems for triangles on graphs adapted to the setting of uncertain edge signs. For each problem we allow the network operator to specify whether they wish to enumerate or count either balanced or unbalanced triangles or both depending on the purpose of their query. The techniques presented in this work assume both are desired.

\section{Baseline Exact Solution}
We introduce a baseline exact solution for Uncertain Triangle Counting and Enumeration in this section. Our baseline involves individually examining each triangle to determine if their sign probabilities indicate they are balanced, unbalanced or unclassified based on adapting established techniques for normal graphs.

The key difference between our problem for both counting and enumeration compared to methods for counting the number of uncertain triangles on a standard uncertain graph (all triangles which exist with probability $\geq t$) is the added search space of requiring the examination of edges which have a low probability. In the standard uncertain graph scenario these edges can be discarded as they do not contribute to a valid output. However, since the semantics of edge probabilities are different when they refer to the probability of an edge being positive, these edges must sill be considered for our problem as they still may contribute to a balanced or unbalanced triangle.

For our baseline solution, we present an algorithm capable of systematically enumerating the uncertain balanced and/or unbalanced triangles in a graph. Our counting algorithm uses the same framework, but instead of outputting the discovered triangles it increments a count to be outputted upon program termination. The complexities of these methods are essentially identical.

\subsubsection{Framework}
Our baseline framework is based on edge iteration with the forward improvement \cite{triangle/latapy, triangle/schank}, a classic technique for triangle enumeration. This method relies on having an ordering metric for all nodes in the network which is defined as:

\begin{definition}
\textbf{Node Order} \cite{triangle/latapy, triangle/schank}: For any two nodes $u$ and $v$, we say $u$ is before $v$ ($u <_{no} v$) using the Node Order metric if:
\begin{enumerate}
    \item $deg(u) < deg(v)$
    \item $id(u) < id(v)$, if $deg(u) = deg(v)$
\end{enumerate}
\end{definition}

Node Order is a way of providing a deterministic order for all nodes in $G$ based on increasing degree, which is used to avoid redundant outputs and to generally speed up the enumeration process.

\IncMargin{0.5em}
\begin{algorithm}[h]
\SetKwInOut{Input}{Input}\SetKwInOut{Output}{Output}
\Input{$G$: Uncertain Signed Graph \\ $t$: Balanced/Unbalanced Probability Threshold}
\Output{{\color{blue} $\btrival{t}{G}$: Uncertain Balanced Triangle List}
\\ {\color{blue} $\utrival{t}{G}$: Uncertain Unbalanced Triangle List}
\\ {\color{violet} $|\btrival{t}{G}|$: Uncertain Balanced Triangle Count}
\\ {\color{violet} $|\utrival{t}{G}|$: Uncertain Unbalanced Triangle Count}}
\For{$u \in V_{no}$}{
\For{$v \in N_{>_{no}}(u)$}{
\ForEach{$w \in N_{>_{no}}(u) \cap N_{>_{no}}(v)$}{
\If{$P_{bal}(\trival{(u, v, w)}) \geq t$}{
{\color{blue} $\btrival{t}{G}.insert(\trival{(u, v, w)})$\;}
{\color{violet} $|\btrival{t}{G}|++$\;}
}
\If{$P_{unbal}(\trival{(u, v, w)}) > t$}{
{\color{blue} $\utrival{t}{G}.insert(\trival{(u, v, w)})$\;}
{\color{violet} $|\utrival{t}{G}|++$\;}
}
}
}
}
\caption{Baseline {\color{blue} UBTE} / {\color{violet} UBTC}}
\label{alg/baselineubtec}
\end{algorithm}
\DecMargin{0.5em}

In essence, we systematically enumerate each triangle and determine which uncertain category (balanced, unbalanced or unclassified) it belongs to. Algorithm \ref{alg/baselineubtec} details the method, with lines in blue being enumeration specific and lines in purple being counting specific. Lines in black are crucial for both methods. Firstly, the nodes are sorted by increasing Node Order (denoted by $V_{no}$) (Line 1) and that list is iterated through. For each node $u$ in $V_{no}$, each neighbour $v$ later in the order (the set $N_{>_{no}}(u)$) is checked. The common neighbours of $u$ and $v$ later than $v$ in the order (i.e. $w \in N_{>_{no}}(u) \cap N_{>_{no}}(v)$) topologically form a triangle with $u$ and $v$ (Lines 2-3). We can find this intersection in $O(deg_{>_{no}}(u) + deg_{>_{no}}(v))$ time assuming we do not have a pre-built hash table for the neighbours of each node or $O(\min\{deg_{>_{no}}(u), deg_{>_{no}}(v)\})$ if we do.

Each triangle is then checked to see if they satisfy the probability requirements to be balanced or unbalanced and the relevant containers are updated accordingly (Lines 4-9). For the pseudocode, we assume discovered triangles are stored in a data structure for enumeration purposes but in practice they can be outputted immediately. This method ensures no duplicate outputs as each triangle is only listed once due to enforcing Node Order.

\subsubsection{Analysis}
\begin{theorem}
    The time complexity of Algorithm \ref{alg/baselineubtec} is $O(|E|^{\frac{3}{2}})$
    \label{thm/baserun}
\end{theorem}
\begin{proof}
As proven by Schank, the time complexity of forward edge iteration is $O(|E|^{\frac{3}{2}})$ \cite{triangle/schank}. Our modifications can each be done in $O(1)$ time meaning the complexity does not change.
\end{proof}

\begin{theorem}
    The memory complexity of Algorithm \ref{alg/baselineubtec} is $O(|E|)$
    \label{thm/basemem}
\end{theorem}
\begin{proof}
We require $O(|E|)$ space to load the network and $O(|V|)$ space to create a neighbour list of all nodes (sorted by Node Order).
\end{proof}

Assuming we output uncertain balanced/unbalanced triangles when we find them (or increase an integer value in the counting variation of the problem which would not change the complexity), this is the maximum memory required to perform a serial implementation. Note that if we opt for the method of speeding up list intersection by pre-hashing neighbourhood lists of nodes this takes an additional $O(|E|)$ memory which does not change the complexity but is a notable additional cost for very large networks.

\section{Improved Exact Solution}
In this section we expand on the baseline solution with improvements targeted at our specific problem setting. 

\subsubsection{Probability Theory}
We first take a look at leveraging probability theory for pruning and early termination of searches of unpromising edges. Before we continue any further, we have an additional critical observation regarding uncertain signed triangles which we will use in the future.

\begin{Lemma}
    Suppose we have a triangle $\Delta$ consisting of the edges $e_1, e_2, e_3$. The probability that $\Delta$ is balanced is equal to exactly the probability a triangle consisting of edges with positive sign probability $(1 - p_{+} (e_1)), (1 - p_{+} (e_2)), (1 - p_{+} (e_3))$ is unbalanced.
\label{lemma/revtri}
\end{Lemma}
\begin{proof}
    Let $P_+(e_1') = 1 - P_+(e_1), P_+(e_2') = 1 - P_+(e_2), P_+(e_3') = 1 - P_+(e_3)$. We can rearrange each equation for $P_+(e_i) = 1 - P_+(e_i')$ appropriately. For a triangle $\Delta'_{(e_1', e_2', e_3')}$, the probability this triangle is unbalanced can be derived as:

\begin{equation*}
\begin{split}
P_{unbal}(\Delta')& = ((1 - P_+(e_1')) * (1 - P_+(e_2')) * (1 - P_+(e_3'))) \\ 
& + (P_+(e_1') * P_+(e_2') * (1 - P_+(e_3'))) \\ 
& + (P_+(e_1') * (1 - P_+(e_2')) * P_+(e_3')) \\ 
& + ((1 - P_+(e_1')) * P_+(e_2') * P_+(e_3')) \\
& = (P_+(e_1) * P_+(e_2) * P_+(e_3)) \\ 
& + ((1 - P_+(e_1)) * (1 - P_+(e_2)) * P_+(e_3)) \\ 
& + ((1 - P_+(e_1)) * P_+(e_2) * (1 - P_+(e_3))) \\
& + (P_+(e_1) * (1 - P_+(e_2)) * (1 - P_+(e_3))) \\ 
& = P_{bal}(\Delta)
\end{split}
\end{equation*}

\end{proof}

Lemma \ref{lemma/revtri} is a useful observation as it means all analysis we do for finding uncertain balanced triangles holds equally true for uncertain unbalanced triangles for the same value of $t$. Essentially, each uncertain balanced triangle can be mapped to an uncertain unbalanced triangle with the `reversed' probabilities.

Naturally, edges with a probability closer to $0.5$ are unlikely to be a part of a triangle with a high probability of being balanced or unbalanced. We first determine the exact upper and lower bounds of a triangle being balanced.

\begin{Lemma}
For any edge $e \in E$, the probability of any triangle containing $e$ being balanced or unbalanced is in the range \\ $[\min\{P_+(e), 1 - P_+(e)\}, \max\{P_+(e), 1 - P_+(e)\}]$
\label{lemma/probrange}
\end{Lemma}
\begin{proof}
    Let us consider the scenario in which we are looking for uncertain balanced triangles.

Let us assume we have a triangle consisting of $e$ and two other edges $e_i$ and $e_j$. Notice that if we have a balanced triangle, flipping the sign of any of its edges results in an unbalanced triangle (and vice versa). Thus, regardless of the probability of the product of the positive sign probabilities of $e_i$ and $e_j$, the maximum probability of a triangle containing $e$ being balanced is $P_+(e)$ in the case that a positive sign would create a balanced triangle or $1 - P_+(e)$ in the case a negative sign would create a balanced triangle.

By extending Lemma \ref{lemma/revtri}, we can infer the same is true for the maximum probability of all unbalanced triangle containing $e$.
\end{proof}

Take the extreme example for a triangle with two edges which are known to be positive ($p = 1$), the probability that the triangle is balanced or unbalanced is dependent exactly on the probability the third edge is positive. The immediate benefit of Lemma \ref{lemma/probrange} is we can prune edges with a probability in the range $[\min\{1 - t, t\}, \max\{1 - t, t\}]$, as they cannot possible form a valid output.

\subsubsection{Absolute Edge Ordering}

As mentioned in the previous section, the forward edge iteration method requires a Node Order to ensure efficiency and avoid redundant outputs in the normal graph setting, where all triangles matter to the output. However for our problem we have seen there is significant pruning power in sign probabilities to avoid unclassified triangles.

\begin{figure}
\centering
\includegraphics[height = 4cm]{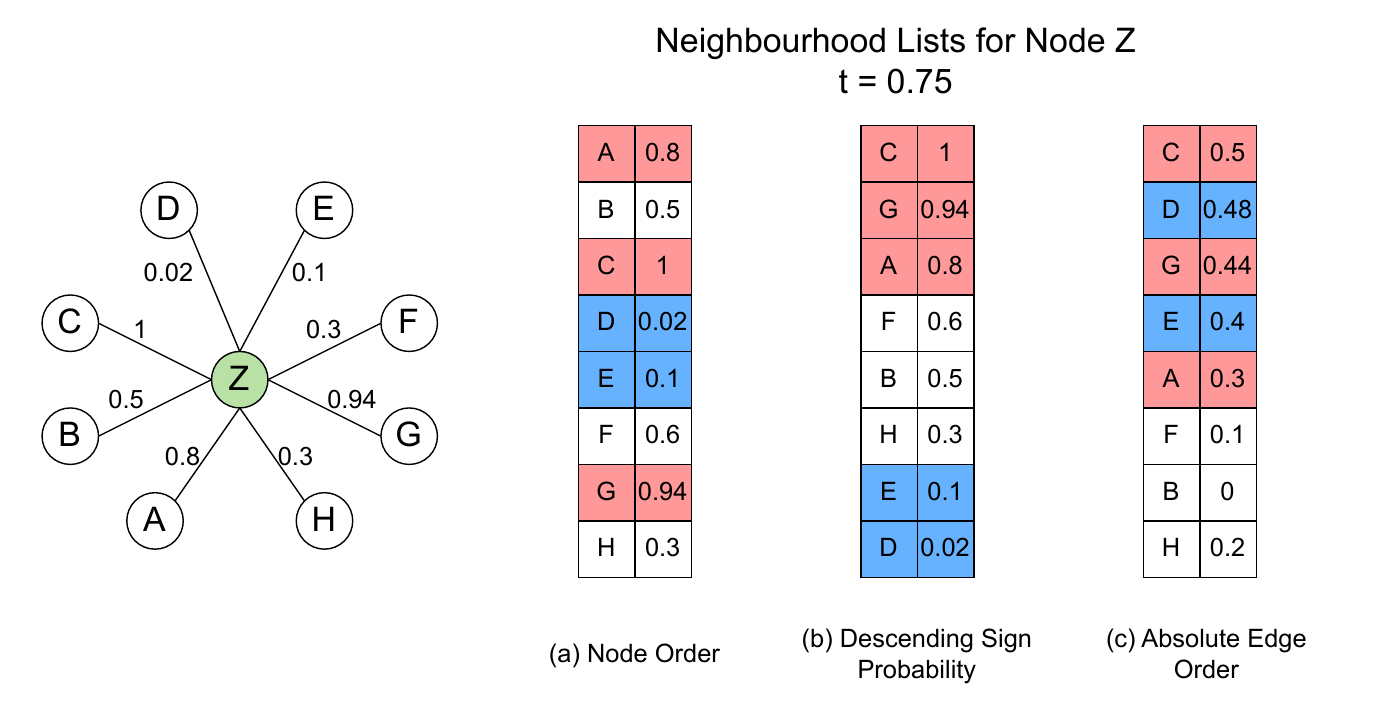}
\caption{An example of different ordering options for the neighbourhood list of node $Z$. (a) is the baseline of sorting by Node Order, (b) is sorted by descending sign probability and (c) is sorted by Absolute Edge Order. Boxes in red and blue are edges which are likely to be positive and negative, respectively, and satisfy the threshold pruning requirement.}
\label{fig/ballist}
\end{figure}

Let us consider the example graph in Figure \ref{fig/ballist}, where we wish to determine efficient ways of ordering the neighbours of node $Z$ which would allow us to prune effectively. On this example where $t = 0.75$, via Lemma \ref{lemma/probrange} we know that we can ignore all edges with positive sign probability in the range $(0.25, 0.75)$. If we were to order the neighbours of $Z$ via Node Order (Fig \ref{fig/ballist} (a)), as per our baseline, we would still have to scan through the list taking $O(deg(Z))$ time.
 
As we have seen, taking into account the information provided by the sign probabilities which is highly useful for pruning, an immediately improved edge order would be based around sorting by positive edge sign probability. Take the example list sorted this way in Fig \ref{fig/ballist} (b), in which the neighbourhood list has been split into three distinct regions. The red and blue cells indicate edges which are likely to be positive and negative, respectively, which cannot be pruned by Lemma \ref{lemma/probrange}. Most notably, the white cells in the list indicate a section which can be pruned when searching any edges associated with $Z$ for our two problems.

Whilst this ordering method defines a clear zone which does not need to be searched, requiring two pointers (starting from the top and bottom) for each list still adds a level of additional implementation as opposed to scanning a sorted list. Additionally, future improvements which we will implement with two pointers would require additional comparisons between the values at the two pointers which adds additional practical runtime. As such we propose a further improved ordering metric specific to our problem setting.

\begin{definition}
\textbf{Absolute Edge Value ($a(e)$):} The Absolute Edge Value of an edge is $a(e) = |P_+(e) - 0.5|$.
\end{definition}

\begin{definition}
\textbf{Absolute Edge Order:} For any two edges $e_1(u_1, v_1)$ and $e_2(u_2, v_2)$, we say $e_1$ is before $e_2$ ($e_1 <_{aeo} e_2$) using the Absolute Edge Order metric if:
\begin{enumerate}
    \item $a(e_1) > a(e_2)$
    \item $id(u_1) < id(u_2)$, if $a(e_1) = a(e_2)$
    \item $id(v_1) < id(v_2)$, if $a(e_1) = a(e_2)$ and $id(u_1) = id(u_2)$
\end{enumerate}
\end{definition}

Absolute Edge Order leverages the fact that the edges most likely to contribute to a valid output are the ones with positive sign probabilities closest to 0 or 1. Note that the range of potential Absolute Edge Values is between 0 and 0.5 as opposed to 0 and 1. Figure \ref{fig/ballist} (c) demonstrates the neighbourhood list of $Z$ using this new sorting method (noting that the equivalent to $t$ for this list is $a(t) = |t - 0.5|$) and illustrates how it combines both edges with high and low positive sign probabilities into an absolute ranking.

In the baseline framework the order we discovered triangles was based on the degree of its nodes (Line 1, Algorithm \ref{alg/baselineubtec}). Given our new metric in our improved method we can instead iterate through the network based on edges (sorted by Absolute Edge Order). This can further reduce the search space when considering the neighbourhood lists of the two nodes connecting a given edge.

\begin{Lemma}
    When finding uncertain balanced or unbalanced triangles containing edge $e(u, v)$, given that we are searching via Absolute Edge Order, we only need to check edges $e'$ connected to $u$ or $v$ which satisfy the requirement $a(P_+(e')) \in \left[ a(\frac{2P_+(e)^2 - 2P_+(e) + t}{4P_+(e)^2-4P_+(e) + 1}) , a(P_+(e))\right]$.
\label{lemma/newcutoff}
\end{Lemma}

\begin{proof}
    Given we search by descending Absolute Edge Order, we know that for $e$ we do not need to recheck any edges with a higher Absolute Edge Value as any valid triangles containing them will have already been discovered due to the order of the search. This provides us the upper bound $|a(P_+(e))|$. 

For the lower bound, let us first consider searching for uncertain balanced triangles. Let $p = P_+(e)$, which is the upper probability of any edge we will search. Suppose there exists a triangle $A$ containing $e$. Assuming that $p \geq 0.5$, the maximum probability for any edge in $T$ (assuming Absolute Edge Order-based search) is $p$, meaning the other two edges have probabilities at most $p$. If we fix the probabilities of one other edge on the triangle to $p$ (along with $e$), we can calculate the minimum probability $x$ ($x$ must $\geq 0.5$ in order to satisfy classification $T_3$) required for the third edge such that $P_{bal}(A) \geq t$ via:
    
\begin{equation*}
\begin{split}
P_{bal}(A) &= t \\
p^2x + (1-x)(1-p)p + (1-x)(1-p)p + x(1-p)^2 & = t \\
p^2x + 2(p - p^2 -px + p^2x) + (1 - 2p + 2p^2)x &= t \\
4p^2x + 4px + x - 2p^2 + 2p & = t \\
x(4p^2 + 4p + 1) &= 2p^2 - 2p + t \\
\frac{2p^2 - 2p + t}{4p^2 - 4p + 1} &= x\\
\end{split}
\end{equation*}

Alternatively, suppose there exists a triangle $B$, the minimum probability for any edge in $B$ is $1 - p$ (due to Absolute Edge Order). This time given that one edge is leaning positive and the other is leaning negative we find the maximum probability $y$, which must be leaning negative as well (i.e. $y \leq 0.5$), in order to form an uncertain balanced triangle with classification $T_1$:

\begin{equation*}
\begin{split}
P_{bal}(B) &= t \\
p(1-p)y + (1-y)(1-p)^2 + (1-y)p^2 + p(1-p)y & = t \\
\frac{-2p^2 + 2p + t - 1}{-4p^2 + 4p - 1} &= y
\end{split}
\end{equation*}

As such, we have identified the range of any valid edge which can form an uncertain balanced triangle containing $e$. Notably, $a(x) = a(y)$, meaning that $a(\frac{2p^2 - 2p + t}{4p^2-4p + 1})$ is sufficient to represent both balanced triangle classes $T_1$ and $T_3$. If the value of $a(\frac{2p^2 - 2p + t}{4p^2-4p + 1})$ is greater than 0.5 or less than 0 that means no triangle containing $e$ can exist as a valid uncertain balanced or unbalanced triangle for this $t$ value.

We can similarly prove that the same holds true for all $p < 0.5$ as well as unbalanced triangles via repeating the above proof or by extending Lemma \ref{lemma/revtri}.
\end{proof}

\subsubsection{Framework}

\IncMargin{0.5em}
\begin{algorithm}[h]
\SetKwInOut{Input}{Input}\SetKwInOut{Output}{Output}
\Input{$G$: Uncertain Signed Graph \\ $t$: Balanced/Unbalanced Probability Threshold}
\Output{{\color{blue} $\btrival{t}{G}$: Uncertain Balanced Triangle List}
\\ {\color{blue} $\utrival{t}{G}$: Uncertain Unbalanced Triangle List}
\\ {\color{violet} $|\btrival{t}{G}|$: Uncertain Balanced Triangle Count}
\\ {\color{violet} $|\utrival{t}{G}|$: Uncertain Unbalanced Triangle Count}}
\For{$e_i(u, v) \in E_{aeo}$}{
\If{$a(e) < a(t)$}{
Break\;
}
$p \leftarrow a(\frac{2p^2 - 2p + t}{4p^2-4p + 1})$\;
\ForEach{$w \in N_{>_{aeo}, >p}(u) \cap N_{>_{aeo}, >p}(v)$}{
\If{$P_{bal}(\trival{(u, v, w)}) \geq t$}{
{\color{blue} $\btrival{t}{G}.insert(\trival{(u, v, w)})$\;}
{\color{violet} $|\btrival{t}{G}|++$\;}
}
\If{$P_{unbal}(\trival{(u, v, w)}) > t$}{
{\color{blue} $\utrival{t}{G}.insert(\trival{(u, v, w)})$\;}
{\color{violet} $|\utrival{t}{G}|++$\;}
}
}
}
\caption{Improved {\color{blue} UBTE} / {\color{violet} UBTC}}
\label{alg/improvedubtec}
\end{algorithm}
\DecMargin{0.5em}

We now provide the framework for our improved Uncertain Balanced and/or Unbalanced Triangle Enumeration/Counting algorithm in Algorithm \ref{alg/improvedubtec}, with once again the enumeration specific operations in blue and the counting specific operations in purple. We go through each edge in $G$, sorted by Absolute Edge Order (Line 1) and terminate our algorithm after an edge $e(u, v)$ appears where $a(e) < a(t)$ due to Lemma \ref{lemma/probrange} (Line 2). For $e$, we determine the lower bound of Absolute Edge Values (denoted as $p$) of edges which could possibly form a valid solution via Lemma \ref{lemma/newcutoff} and find all common neighbours of $u$ and $v$ connected by edges which are after $e$ in the Absolute Edge Order but also have an Absolute Edge Value $\geq p$ denoted by the neighbourhoods $N_{>_{aeo}, >p}(u), N_{>_{aeo}, >p}(v)$ respectively (Lines 4-5). We then update data containers if the resulting triangle is uncertain balanced or unbalanced (Lines 6-11).

\subsubsection{Adaption to Top-$k$ Search}
Another highly valid query is finding the top-$k$ balanced and unbalance triangle on the network or locally for each node. Our framework can easily be adapted to this problem, given that out algorithm targets Uncertain Triangles with the highest and lowest balance/unbalance probabilities. The only requirement is that once $k$ triangles have been found, the required $t$ value can be adjusted to $t_{new}$ (equivalent to the lowest balanced/highest unbalanced probability for a triangle in the current output set) to further prune the space via $a(t_{new})$. The complexity analysis is unchanged.

\subsubsection{Analysis}

\begin{theorem}
    The time cost of Algorithm \ref{alg/improvedubtec} is \\ $O(\sum_{e(u, v) \in E_{a(e)>a(t)}}(deg_{>_{aeo}, >p}(u)+deg_{>_{aeo},>p}(v)))$.
\label{thm/imprTime}
\end{theorem}

\begin{proof}
    We only iterate through edges in which their absolute value is $\geq |t - 0.5|$, which we denote as the set $E_{a(e)>a(t)}$. For each edge $e(u, v)$ which satisfies that requirements, we find the common neighbours of $u$ and $v$ which satisfy the search range of Lemma \ref{lemma/newcutoff} (denoted by $N_{>_{aeo},>p}$ thus having a degree of $deg_{>_{aeo}, >p}$ each). If we do not pre-hash each neighbourhood for each edge this takes $O(deg_{>_{aeo}, >p}(u)+deg_{>_{aeo}, >p}(v))$ time or \\$O(\min\{deg_{>_{aeo}, >p}(u), deg_{>_{aeo}, >p}(v)\})$ time if we do.
\end{proof}

The time cost of Algorithm \ref{alg/improvedubtec} is no worse than the baseline (Algorithm \ref{alg/baselineubtec}), as the worst case in which the improved solution prunes no additional space (due to either a significant number of certain edges or a low $t$ value) should result in theoretically the same search time complexity. Given the nature of the problem, it is assumed the network is stored in a sorted by Absolute Edge Order in order to facilitate multiple queries on the same dataset. If not pre-sorted, we can use any sorting method to sort into Absolute Edge Order, taking $O\sum_{e \in E}(log(deg(e))$ time.

\begin{theorem}
    The memory cost of Algorithm \ref{alg/improvedubtec} is $O(|E|)$.
    \label{thm/imprmem}
\end{theorem}

\begin{proof}
    We require $O(|E|)$ space to load the network and create a neighbour list of all nodes (ordered by Absolute Edge Order)
\end{proof}

The memory complexity of the improved solution is the same as the baseline, though in practice it may take slightly more memory due to the introduction of Absolute Edge Value information for each edge.

\section{Sampling-Based Approximate Solutions for Counting}
Whilst the exact solutions for Uncertain Balanced/Unbalanced Triangle Counting are efficient they are still slow for huge networks. Often for these giant networks there is minimal need to know the exact count and more often a close approximation is sufficient for analytical purposes. In this section we explore unbiased methods of sampling to approximate the Uncertain Balanced or Unbalanced Counts.

\subsection{Vertex-Based Sampling}
\begin{algorithm}
\SetKwInOut{Input}{Input}\SetKwInOut{Output}{Output}
\Input{$G$: Uncertain Signed Graph \\ $t$: Balanced/Unbalanced Probability Threshold \\ $k$: Number of Samples}
\Output{$\tilde{|\btrival{t}{G}|}$: Approximate Uncertain Balanced Triangle Count \\
$\tilde{|\utrival{t}{G}|}$: Approximate Uncertain Unbalanced Triangle Count}
$i \leftarrow 0$\;
\While{$i \leq k$}{
$v \leftarrow$ Uniform Random Sample from $V$\;
$|\btri{t}(v)|, |\utri{t}(v)| \leftarrow VertexLocalCount(v)$\;
$|\tilde{\btrival{t}{G}}| \leftarrow \frac{i*|\tilde{\btrival{t}{G}}| + \frac{|\btri{t}(v)|*|V|}{3}}{i + 1}$\;
$|\tilde{\utrival{t}{G}}| \leftarrow \frac{i*|\tilde{\utrival{t}{G}}| + \frac{|\utri{t}(v)|*|V|}{3}}{i + 1}$\;
$i++$\;
}
\caption{Baseline Vertex Sampling Algorithm}
\label{alg/vsamp}
\end{algorithm}

Algorithm \ref{alg/vsamp} details the vertex sampling framework, which is a standard framework. For $k$ samples, the algorithm samples $k$ nodes from $G$ and for each determines the local uncertain balanced and unbalanced triangle counts ($|\btri{t}(v)|, |\utri{t}(v)|$) relating to that node (Lines 1-4). $\frac{|\btri{t}(v)|*|V|}{3}$ is a method of extrapolating the local balanced count to the entire graph (assumes that each node has the derived local count) and a running average is calculated for the balanced and unbalanced estimators (Lines 5-6).
We used a simple baseline vertex sampling framework (Algorithm \ref{alg/vsamp} which samples a random vertex uniformly $k$ times (without replacement) and determines the local uncertain balanced and unbalanced counts ($|\btri{t}(v)|, |\utri{t}(v)|$) for each vertex. $\frac{|\btri{t}(v)|*|V|}{3}$ extrapolated the local balanced count to the entire graph (assumes that each node has the derived local count) and a running average is calculated, with the $\frac{|\utri{t}(v)|*|V|}{3}$ used for uncertain counts.

This framework determines the local uncertain balanced and unbalanced counts by finding all triangles containing $v$. Given we have already proposed techniques geared towards our specific problem setting, we can adopt the improved solutions when appropriate.

\IncMargin{0.5em}
\begin{algorithm}[h]
\SetKwInOut{Input}{Input}\SetKwInOut{Output}{Output}
\Input{$G$: Uncertain Signed Graph \\ $t$: Balanced/Unbalanced Probability Threshold \\ $q$: Query Vertex}
\Output{$|\btri{t}(q)|$: Local Uncertain Balanced Triangle Count
\\ $|\utri{t}(q)|$: Local Uncertain Unbalanced Triangle Count}
\For{$u \in N_{>_{aeo}}(q)$}{
\If{$a(P_+(e(q, u)))< a(t)$}{
\Return{$|\btri{t}(v)|$, $|\utri{t}(v)|$}\;
}
$p \leftarrow a(\frac{2p^2 - 2p + t}{4p^2-4p + 1})$\;
\ForEach{$v \in N_{>p}(q) \cap N_{>_{aeo},>t}(u)$}{
\If{$P_{bal}(\Delta_{(q, u, v)}) \geq t$}{
$|\btri{t}(q)|++$\;
}
\If{$P_{unbal}(\Delta_{(q, u, v)}) > t$}{
$|\utri{t}(q)|++$\;
}
}
}
\caption{Improved Vertex Local Search}
\label{alg/improvedvsamp}
\end{algorithm}
\DecMargin{0.5em}

Algorithm \ref{alg/vsamp} details the vertex sampling framework, which is a standard framework. For $k$ samples, the algorithm samples $k$ nodes from $G$ and for each determines the local uncertain balanced and unbalanced triangle counts ($|\btri{t}(v)|, |\utri{t}(v)|$) relating to that node (Lines 1-4). $\frac{|\btri{t}(v)|*|V|}{3}$ is a method of extrapolating the local balanced count to the entire graph (assumes that each node has the derived local count) and a running average is calculated for the balanced and unbalanced estimators (Lines 5-6).

To determine the local uncertain balanced and unbalanced counts we can utilise the structure of finding all triangles containing $v$ and checking each triangle much like our exact baseline. However, given we have already examined techniques geared towards our specific problem setting, we can adopt those when appropriate.

\begin{theorem}
    The time complexity of Algorithm \ref{alg/improvedvsamp} is $O(deg_{>t}(q) * (deg_{>t}(q) + maxdeg_{>t}))$, where $maxdeg_{>t}$ is the maximum number of edges connected to a node in $G$ with an Absolute Edge Value $\geq a(t)$.
\label{thm/vSampTime}
\end{theorem}

\begin{proof}
    For each query vertex $q$, we check each of it's neighbour $u$ for triangles containing $e(q, u)$ (a total of $O(deg_{>t}(q)$ time). Assuming we do not pre-hash neighbour lists the cost of finding triangles with $e(q, u)$ is $O(deg_{>p}(q) + deg_{>_{aeo}, >t}(u))$ where $deg_{>_{aeo}, >t}(y) \leq maxdeg_{>t}$. If we do pre-hash, the cost of finding triangles containing $e(q, u)$ is $O(\min \{ deg_{>p}(q), deg_{>_{aeo},>t}(u) \})$.
\end{proof}

The worst case time cost of our vertex-based sampling approach (assuming $k$ samples) can be simplified to $O(k*maxdeg^2)$. However in reality our optimizations and pruning heuristics significantly reduce the runtime assuming a reasonable $t$ value. Our method also requires no additional memory aside from a few negligible counters to track values.

\begin{theorem}
    Algorithm \ref{alg/improvedvsamp} is unbiased.
\label{thm/vSampBias}
\end{theorem}

\begin{proof}
    Let us first show Vertex-based Balanced Triangle Counting is Unbiased.

Let the true Uncertain Balanced Triangle Count be denoted by $|B\Delta^t_G| = \frac{1}{3}\sum_{v \in V} (|B\Delta^t(v)|)$.

Let the estimated Uncertain Balanced Triangle Count be denoted by $|\tilde{\btrival{t}{G}}| = \frac{|V|}{3k}\sum_{i = 1}^{k} (|B\Delta(v_i)|)$.

\begin{equation*}
\begin{split}
\mathbb{E}[\tilde{\btrival{t}{G}}] &= \mathbb{E} \left[ \frac{|V|}{3k}\sum_{i = 1}^{k} (|B\Delta(v_i)|) \right]\\
&= \frac{|V|}{3} \mathbb{E} \left[ \frac{1}{k} \sum_{i = 1}^{k} (|B\Delta(v_i)|) \right]\\
&= \frac{|V|}{3} \frac{1}{|V|} \sum_{v \in V} (|B\Delta^t(v)|)\\
&= \frac{1}{3} \sum_{v \in V} (|B\Delta^t(v)|)\\
&= |\btrival{t}{G}|
\end{split}
\end{equation*}

The proof that Vertex-based Unbalanced Triangle Counting is unbiased follows the exact same structure, only replacing $\btrival{t}{G}$ with $\utrival{t}{G}$.
\end{proof}

\subsection{Edge-Based Sampling}
\begin{algorithm}
\SetKwInOut{Input}{Input}\SetKwInOut{Output}{Output}
\Input{$G$: Uncertain Signed Graph \\ $t$: Balanced/Unbalanced Probability Threshold \\ $k$: Number of Samples}
\Output{$\tilde{|\btrival{t}{G}|}$: Approximate Uncertain Balanced Triangle Count \\
$\tilde{|\utrival{t}{G}|}$: Approximate Uncertain Unbalanced Triangle Count}
$i \leftarrow 0$\;
\While{$i \leq k$}{
$e \leftarrow$ Uniform Random Sample from $E$\;
$|\btri{t}(e)|, |\utri{t}(e)| \leftarrow EdgeLocalCount(e)$\;
$|\tilde{\btrival{t}{G}}| \leftarrow \frac{i*|\tilde{\btrival{t}{G}}| + \frac{|\btri{t}(e)|*|E|}{3}}{i + 1}$\;
$|\tilde{\utrival{t}{G}}| \leftarrow \frac{i*|\tilde{\utrival{t}{G}}| + \frac{|\utri{t}(e)|*|E|}{3}}{i + 1}$\;
$i++$\;
}
\caption{Baseline Edge Sampling Algorithm}
\label{alg/esamp}
\end{algorithm}

Edge-based sampling is essentially a choice to further reduce the search space of each local sample in comparison to the vertex-based approach. We adapt our vertex-sampling framework to this variation based on a similar edge-based sampling framework which gets the local count for each edge then extrapolates and averages of $k$ samples (Algorithm \ref{alg/esamp}).

\IncMargin{0.5em}
\begin{algorithm}[h]
\SetKwInOut{Input}{Input}\SetKwInOut{Output}{Output}
\Input{$G$: Uncertain Signed Graph \\ $t$: Balanced/Unbalanced Probability Threshold \\ $e_q(u, v)$: Query Edge}
\Output{$|\btri{t}(e_q)|$: Local Uncertain Balanced Triangle Count
\\ $|\utri{t}(e_q)|$: Local Uncertain Unbalanced Triangle Count}
\If{$a(P_+(e(q, u))) < a(t)$}{
\Return{$|\btri{t}(e_q)| = 0, |\utri{t}(e_q)| = 0$}\;
}
\ForEach{$w \in N_{>t}(u) \cap N_{>t}(v)$}{
\If{$P_{bal}(\Delta_{(u, v, w)}) \geq t$}{
$|\btri{t}(e_q)|++$\;
}
\If{$P_{unbal}(\Delta_{(u, v, w)}) > t$}{
$|\utri{t}(e_q)|++$\;
}
}
\caption{Improved Edge Local Search}
\label{alg/improvedesamp}
\end{algorithm}
\DecMargin{0.5em}

The details of deriving the local counts of an edge is detailed in Algorithm \ref{alg/improvedesamp}. We can immediately terminate the search with the current count if the edge itself is unable to be a part of any relevant result due to its own probability via Lemma \ref{lemma/probrange} (Line 1-2). When finding the common neighbours of $u$ and $v$, we can once again utilise $a(t)$ as a point of early termination as we search through $N(u)$ (assuming $deg(u) < deg(v)$) as it is sorted via Absolute Edge Order. We then only check probabilities of triangles which have passed all these conditions (Lines 3-7).

\begin{theorem}
    The time complexity of Algorithm \ref{alg/improvedesamp} is \\ $O(deg_{>t}(u) + deg_{>t}(v))$ assuming neighbour lists are pre-hashed and $O(\min\{deg_{>t}(u), deg_{>t}(v)\})$ otherwise.
    \label{thm/eSampTime}
\end{theorem}

\begin{proof}
    This proof follows naturally from Theorem \ref{thm/imprTime} and \ref{thm/vSampTime} with the only difference being the search space in the neighbour list $N(u)$ (assuming $deg_{>t}(u) < deg_{>t}(v)$) cannot be bounded by $p$ and instead must be bounded by $a(t)$.
\end{proof}

Again, our sampling approach requires no additional memory complexity aside from a few counters to keep track of value.

\begin{theorem}
    Algorithm \ref{alg/improvedesamp} is unbiased.
\label{thm/eSampBias}
\end{theorem}

\begin{proof}
    Let us first show Edge-based Balanced Triangle Counting is Unbiased.

Let the true Uncertain Balanced Triangle Count be denoted by $|B\Delta^t_G| = \frac{1}{3}\sum_{e \in E} (|B\Delta^t(e)|)$.

Let the estimated Uncertain Balanced Triangle Count be denoted by $|\tilde{\btrival{t}{G}}| = \frac{|E|}{3k}\sum_{i = 1}^{k} (|B\Delta(e_i)|)$.

\begin{equation*}
\begin{split}
\mathbb{E}[\tilde{\btrival{t}{G}}] &= \mathbb{E} \left[ \frac{|E|}{3k}\sum_{i = 1}^{k} (|B\Delta(e_i)|) \right]\\
&= \frac{|E|}{3} \mathbb{E} \left[ \frac{1}{k} \sum_{i = 1}^{k} (|B\Delta(e_i)|) \right]\\
&= \frac{|E|}{3} \frac{1}{|E|} \sum_{e \in E} (|B\Delta^t(e)|)\\
&= \frac{1}{3} \sum_{e \in E} (|B\Delta^t(e)|)\\
&= |\btrival{t}{G}|
\end{split}
\end{equation*}

The proof that Edge-based Unbalanced Triangle Counting is unbiased follows the exact same structure, only replacing $\btrival{t}{G}$ with $\utrival{t}{G}$.
\end{proof}

\section{Experiments}
In this section we explore the efficiency (and when appropriate effectiveness) of our proposed solutions. Our implementation is available online \footnote{https://anonymous.4open.science/r/UncertainBalancedTriangles-6086}. Our techniques are implemented in Standard C++11 on an Intel Xeon Gold 6248 CPU @ 2.50GHz with 1.5TB of memory. We note that, given the nature of floating point arithmetic, our implementation cannot guarantee perfect accuracy for probability values with a long number of significant figures. That being said, particularly in regards to effectiveness, our implementation fairly evaluates our contribution. 

\subsection{Datasets}
\begin{table*}
\begin{center}
\caption{Dataset Information}
\begin{tabular}{ |c|c|c|c|c|c|c|} 
\hline
\textbf{Dataset} & $|V|$ & $|E|$ & $|\Delta|$ & Average $deg$ & Maximum $deg$ & Dist\\ 
\hline
Hamster (HA) & 1,858 & 12,532 & $1.67 \times 10^4$ & 8.49 & 272 & Uniform\\
\hline
Brightkite (BK) & 58,228 & 214,078 & $4.95 \times 10^5$ & 7.35 & 1,134 & Uniform\\
\hline
Slashdot (SD) & 79,116 & 467,730 & $5.38 \times 10^5$ & 13.03 & 2,543 & Signs\\
\hline
Epinions (EP) & 131,579 & 840,798 & $8.53 \times 10^6$ & 12.76 & 3,622 & Signs \\
\hline
LiveMocha (LM) & 104,103 & 2,193,083 & $3.36 \times 10^6$ & 42.13 & 2,980 & Uniform \\
\hline
Flixster (FX) & 2,523,386 & 7,918,801 & $7.89 \times 10^6$ & 6.276 & 1,474 & Uniform\\
\hline
Skitter (SK) & 1,696,415 & 11,095,298 & $2.88 \times 10^7$ & 13.08 & 35,455 & Beta\\
\hline
Petster (PE) & 623,754 & 13,991,746 & $2.69 \times 10^8$ & 44.86 & 80,634 & Beta\\
\hline
Flickr (FL) & 1,715,253 & 15,551,249 & $1.08 \times 10^8$ & 18.13 & 27,224 & Uniform \\
\hline
LiveJournal (LJ) & 3,997,962 & 34,681,189 & $1.78 \times 10^8$ & 17.35 & 14,815 & Beta \\
\hline
Orkut (OR) & 3,073,441 & 117,184,899 & $6.27 \times 10^8$ & 76.28 & 33,313 & Uniform \\
\hline
UK Domain (UK) & 18,483,186 & 261,787,258 & $4.45 \times 10^9$ & 28.33 & 194,955 & Uniform \\
\hline
\end{tabular}
\label{table/datasets}
\end{center}
\end{table*}

We conduct experiments on 12 different datasets, whose details can be found in Table \ref{table/datasets}. We examine various distributions of probabilities in order to determine the effectiveness of our algorithms under different environments. All base datasets are available from Stanford Snap \footnote{\url{http://snap.stanford.edu/data/index.html}} and Konect \footnote{\url{http://konect.cc/networks/}} however some may have been modified to remove self and multi-edges.

Slashdot (SD) and Epinions (EP) are two real-world signed networks and in order to capture the known distribution of signs we use a beta distribution for each edge depending on the sign such that the most likely probability to be assigned to the edge is $1$ if the sign is positive and $0$ if the sign is negative. For our other datasets, we assign either a uniform (between 0 and 1) or beta ($\alpha = 0.5, \beta = 0.5$) distribution. The choice of beta distribution biases the probabilities of edges to be closer to 0 and 1, which simulates a network with higher confidence in edge sign prediction.

\subsection{Performance of Exact Counting Algorithms}
We evaluate the performance of both our \textbf{Baseline} (Section 4) and \textbf{Improved} (Section 5) exact solutions in regards to counting the number uncertain balanced and unbalanced triangles in a network. We opt for pre-hashed neighbour lists for increased efficiency and the data is pre-sorted to simulate an environment in which these queries are expected to be performed repeatedly. In the interest of avoiding redundancy, we do not add individual experiments for enumeration as the runtime is near identical to the counting results.

\begin{figure}
\centering
\includegraphics[width = 1\columnwidth]{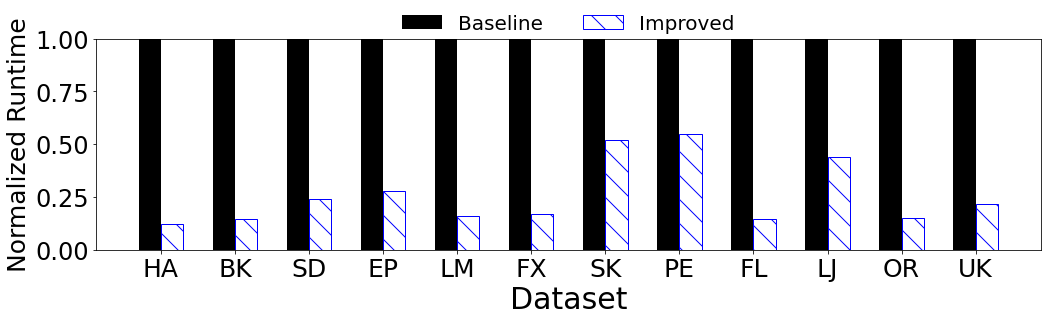}
\caption{Comparing the runtime of the Baseline with the Improved solution for all networks, for $t = 0.8$}
\label{fig/exactrun}
\end{figure}

To begin we examine the runtime of our improved solution in comparison to the proposed baseline for all our datasets for a threshold value of $t = 0.8$, shown in Figure \ref{fig/exactrun}. Due to the wide range of dataset sizes and subsequent runtimes, we normalise our results such that the baseline is 1 and the improved solution is a proportion of 1. In general, it is evident that our proposed improvements does lead to significant decrease in runtime which is in line with out expectations and analysis.

\begin{figure}
\centering
\includegraphics[height = 5cm]{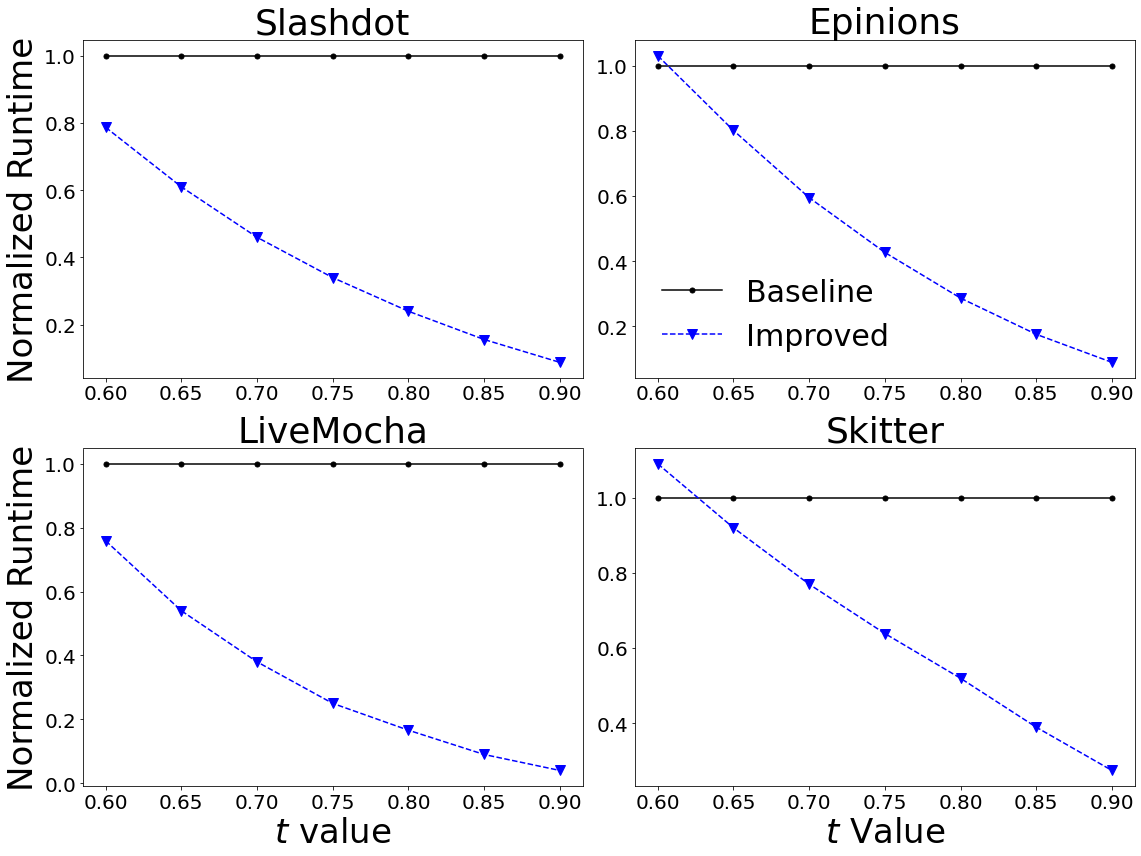}
\caption{Examining the change in runtime as $t$ increases on four different datasets.}
\label{fig/exactshift}
\end{figure}

Of course, the performance of our method is significantly influenced by the selection of the input $t$. Figure \ref{fig/exactshift} details the change in performance of our solution for an increasing value of $t$. In general it is clear there is a noticeable decrease in runtime as $t$ approaches $1$. We notice that our method at times is outperformed by the baseline for low $t$ values. This is due to the increased number of constant time procedures in implementation. In general, it is expected that a realistic selection for $t$ would be on the higher side as a triangle with an unconfident likelihood of being balanced or unbalanced is not useful for practical purposes. As such, for practical queries our results highlight a significant improvement upon the proposed baseline.

From our experiments, we notice that there appears to be a notable difference in runtime at $t = 0.8$ between networks with uniform and beta distributions for exact counting, which is unsurprising. To provide more detailed analysis, we fix the base network and re-assign sign probabilities from various distributions to examine the effect.

\begin{figure}
\centering
\includegraphics[height = 3cm]{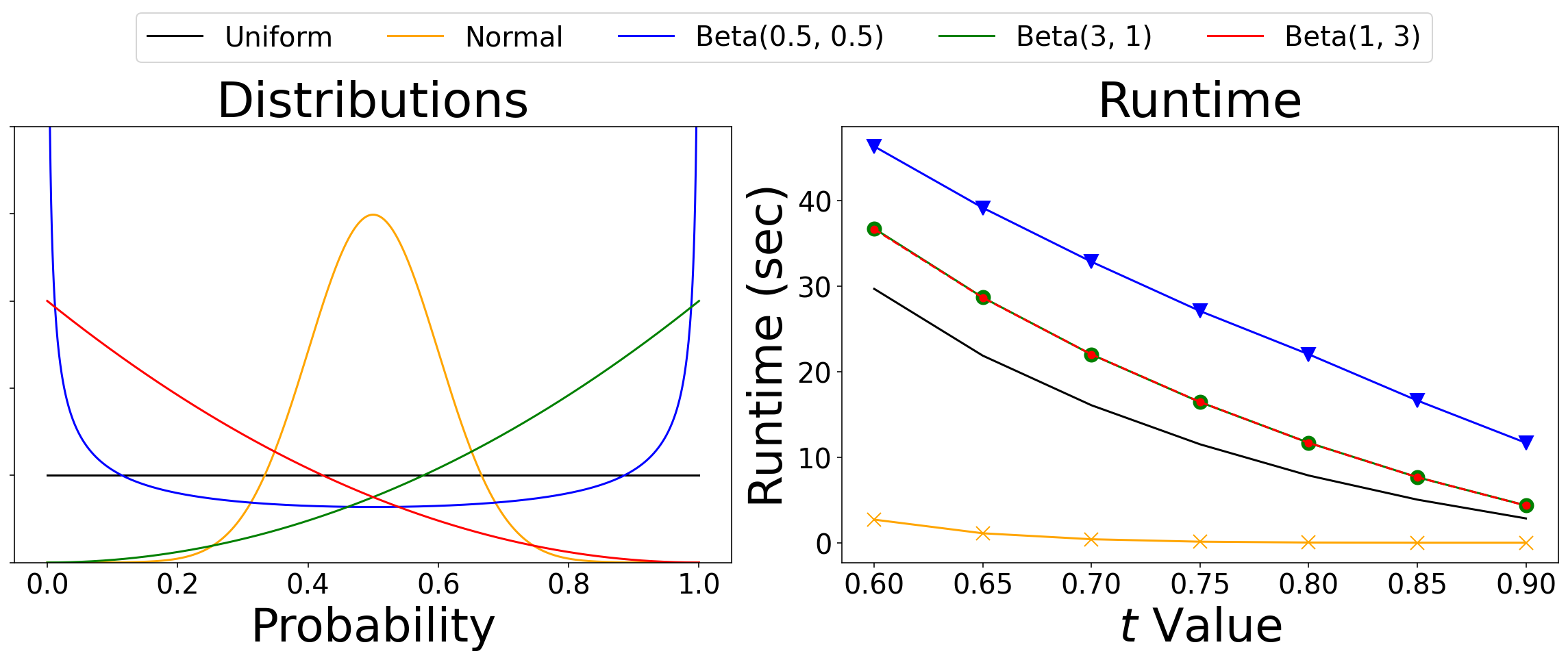}
\caption{An examination of the runtime of our improved solution for different edge sign distributions on the Skitter Dataset. The left figure displays the PDFs of the distributions and the right displays runtime as $t$ increases.}
\label{fig/exactdist}
\end{figure}

We examine five distributions and their effects on runtime in Figure \ref{fig/exactdist}, with the baseline dataset being Skitter. The choices of distribution are Uniform, Normal and three Beta distributions with parameters $(\alpha = 0.5, \beta = 0.5)$, $(\alpha = 3, \beta = 1)$ and $(\alpha = 1, \beta = 3)$. On the left plot of the figure, we can visualise the probability density functions of the distributions and observe that they cover a reasonable variety.

In regards to effects on runtime, on the right of Figure \ref{fig/exactdist} we plot the runtime of the improved algorithm as $t$ increases. The rule of thumb is that the more the distribution biases towards producing values close to $0$ and $1$, the longer the algorithm will take as less of the search space can be pruned. This is most evident by visualising how a normally distributed set of edge sign probabilities will result in an extremely fast execution time. The observed results unsurprising given our previous analysis of the improved algorithm.

\subsection{Performance of Approximate Counting Algorithms}
We now turn our attention to the sampling based methods (established in Section 6) in regards to evaluating their performance. The methods we evaluate are:

\begin{itemize}
    \item \textbf{vBase:} Vertex-based sampling with no improvements.
    \item \textbf{vImpr:} Vertex-based sampling with improved local search (Section 6.1)
    \item \textbf{eBase:} Edge-based sampling with no improvements
    \item \textbf{eImpr:} Edge-based sampling with improved local search (Section 6.2)
\end{itemize}

\begin{figure}
\centering
\includegraphics[width = 1\columnwidth]{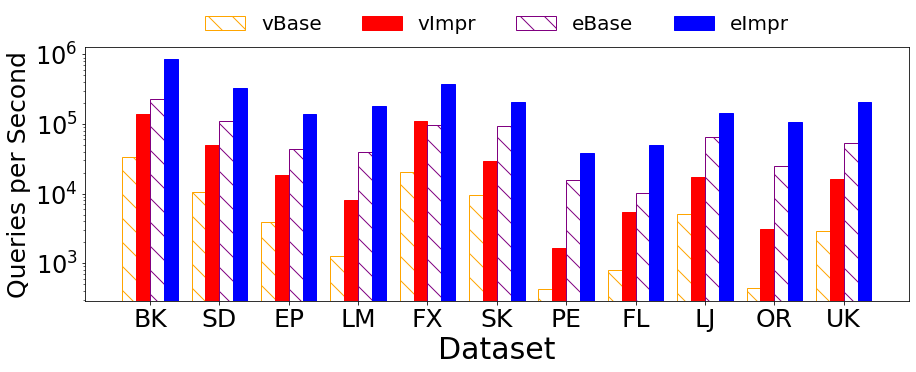}
\caption{Comparing the queries per second of the four sampling approaches for the tested networks, for $t = 0.8$}
\label{fig/sampqps}
\end{figure}

The efficiency of our sampling methods is examined in Figure \ref{fig/sampqps} (note the y-axis is logarithmic in scale), which details the Queries per Second of each sampling method on a number of our datasets ($t = 0.8$). We omit HA from our sampling experiments as it is too small to produce any meaningful insights.  

\begin{table}
\begin{center}
\caption{Runtimes (sec) of Each Method. Sampling assumes 10,000 samples}
\begin{tabular}{ |c|c|c|c|c|c|c|} 
\hline
{\bf Dataset} & bExact & iExact & vBase & vImpr & eBase & eImpr \\
\hline
BK & 0.493 & 0.073 & 0.293 & 0.072 & 0.043 & 0.011 \\
\hline
SD & 1.664 & 0.260 & 0.944 & 0.198 & 0.091 & 0.030 \\
\hline
EP & 5.503 & 1.581 & 2.550 & 0.544 & 0.235 & 0.071 \\
\hline
LM & 16.482 & 2.748 & 8.000 & 1.256 & 0.253 & 0.055 \\
\hline
FX & 31.256 & 5.302 & 0.488 & 0.091 & 0.102 & 0.026 \\
\hline
SK & 42.374 & 22.045 & 1.038 & 0.234 & 0.108 & 0.028 \\
\hline
PE & 219.343 & 122.538 & 23.753 & 6.042 & 0.647 & 0.265 \\
\hline
FL & 447.957 & 66.283 & 12.36 & 1.856 & 0.982 & 0.201 \\
\hline
LJ & 222.844 & 99.849 & 1.987 & 0.584 & 0.154 & 0.070 \\
\hline
OR & 1645.686 & 246.205 & 22.634 & 3.231 & 0.402 & 0.092 \\
\hline
UK & 1268.744 & 277.160 & 3.503 & 0.616 & 0.184 & 0.049 \\
\hline
\end{tabular}
\label{table/runtimes}
\end{center}
\end{table}

We compare the runtime of our exact solutions to sampling methods in Table \ref{table/runtimes}. For sampling methods, we report the runtime after 10,000 samples, which via Figure \ref{fig/sampacc} we see already achieves good accuracy. bExact and iExact represent the baseline and improved exact solutions respectively. We observe that on the larger datasets, the fastest sampling method (eSamp) achieves at least a 3 magnitude speed-up over the fastest exact method (iExact) which makes it extremely appealing for large systems.

In general, all datasets show that the improved sampling solutions dramatically outperforms the baseline solution for both vertex and edge-centric methods. Furthermore, vertex sampling is slower than edge sampling which is also to be expected as each vertex sample can be though of as equivalent to a number of edge samples. The results also demonstrate the scalability of sampling as a method for dealing with large networks, as the number of queries per second is not affected by the number of nodes or edges but more so via average degree.

\begin{figure}
\centering
\includegraphics[height = 3cm]{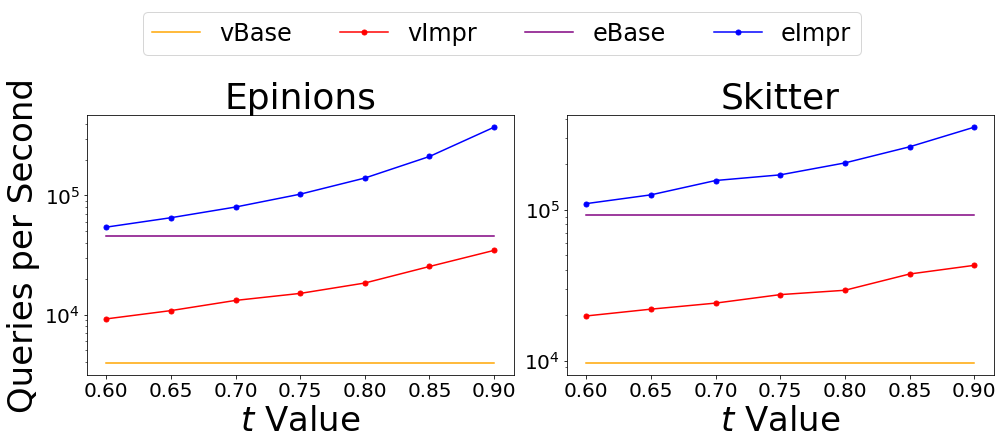}
\caption{An examination of the change in queries per second of sampling algorithms as $t$ increases.}
\label{fig/samprunt}
\end{figure}

The efficiency of our algorithms improves as $t$ increases as visualised by Figure \ref{fig/samprunt}. On Epinions and Skitter we notice significant improvement as $t$ increases on our proposed solutions in both the vertex and edge-centric method, whilst the baseline remains unchanged, which is to be expected. Once again, we highlight the y-axis is logarithmic.

\begin{figure}
\centering
\includegraphics[height = 5.5cm]{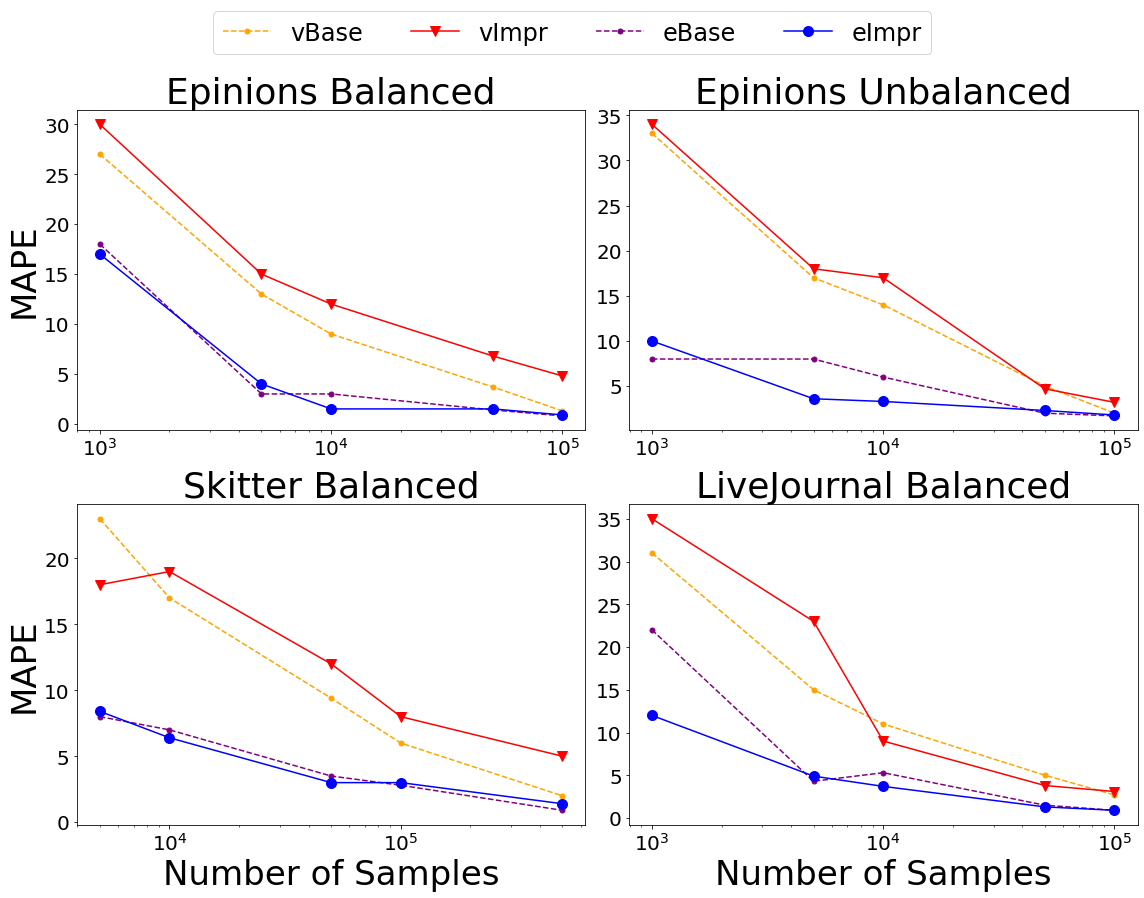}
\caption{The Mean Absolute Percentage Error (MAPE) of sampling methods as the number of samples increases ($t = 0.8$)}
\label{fig/sampacc}
\end{figure}

Figure \ref{fig/sampacc} details the Mean Absolute Percentage Error (MAPE) as the number of samples increases (a MAPE score of 0 means a perfect estimation) on three datasets (with both the uncertain balanced and unbalanced counts being observed on Epinions), averaged after 100 runs. We observe that all methods converge towards the true count, with both baseline and improved methods of the same type doing so at approximately the same rate. Our methods are attractive given the minimal time cost associated with sampling.

On the datasets tested, we found that edge-based sampling never converged slower than vertex-based sampling making \textbf{eImpr} the clear choice due to it's efficiency. As such, that method is our recommended choice if sampling is the desired direction of the network operator.

\section{Conclusion and Future Work}
In this work, we adapted the concept of balanced and unbalanced triangles to a setting in which edge signs exist with uncertainty and explored the problem of enumerating and counting the number of these triangles on such a system. We provided frameworks for both exact counting and enumeration as well as a direction for sampling the counts. Our experiments support the theoretical analysis of our algorithms, demonstrating the feasibility of our solutions.

The future direction this work provides is also interesting, with the problem of triangle counting lending itself in the direction of distributed or streaming algorithms to further increase the scalability of the problem. Another direction that may be promising is that of a distributed network perspective which would allow for larger network evaluation which could utilise exact approach on a smaller space to reach an exact count in a faster time frame. As such, a distributed framework (in particular the problem of forming suitable partitions) for the problem of counting uncertain balanced and unbalanced triangles is also of significant interest for future work.

As mentioned through the paper, structural balance theory is foundational to a large family of signed graph problems, with one particular area of interest being community detection. Notably, with these triangle counting and enumeration techniques in place a natural extension would be to formulate a suitable truss-based model in order to capture close communities with strong balanced support. With the foundations provided by this work, these problems can similarly be extended into the realistic problem space of networks with uncertain sign probabilities. Future attention in this area is critical in expanding the real-world usability of signed networks in order to provide algorithms and frameworks that are more flexible and informative than fixing signs based on uncertain methods.

\bibliographystyle{ACM-Reference-Format}
\bibliography{references}

\end{document}